\documentclass[10pt,final,twocolumn,journal]{IEEEtran}

\usepackage[cmex10]{amsmath}		
\usepackage{amssymb}
\usepackage{amsthm}
\usepackage{amsfonts}
\usepackage{enumerate}
\usepackage{array}
\usepackage{graphicx}
\usepackage[caption=false]{subfig}
\usepackage{wrapfig}
\usepackage{float}
\usepackage{tikz}
\usetikzlibrary{arrows}
\usetikzlibrary{shapes}
\usetikzlibrary{shapes.misc}
\usetikzlibrary{fit}
\usetikzlibrary{calc}
\usetikzlibrary{intersections}
\usetikzlibrary{decorations.pathmorphing}
\usetikzlibrary{decorations.pathreplacing}

\newtheorem{theorem}{Theorem}

\newtheorem{lemma}{Lemma}

\makeatletter
\DeclareRobustCommand*{\bfseries}{%
  \not@math@alphabet\bfseries\mathbf
  \fontseries\bfdefault\selectfont
  \boldmath
}
\makeatother

\makeatletter
\def\blfootnote{\xdef\@thefnmark{}\@footnotetext}
\makeatother

\begin{document}

\title{Repairable Fountain Codes}

\author{%
  Megasthenis~Asteris,~\IEEEmembership{Student~Member,~IEEE,} %
  Alexandros~G.~Dimakis,~\IEEEmembership{Member,~IEEE}%
  \IEEEcompsocitemizethanks{\IEEEcompsocthanksitem%
  M. Asteris and A. G. Dimakis are with the University of Texas at Austin.%
  }%
  \thanks{This research was supported in part by NSF Career Grant CCF-1055099
and research gifts by Intel and Microsoft Research.
A preliminary version of this work appeared in \cite{Asteris:ISIT12}.}%
}

\maketitle
\begin{abstract}
We introduce a new family of Fountain codes that are systematic and also have
sparse parities.
Given an input of $k$ symbols, our codes produce an unbounded number of output
symbols, generating each parity independently by linearly combining a
logarithmic number of randomly selected input symbols.
The construction guarantees that for any $\epsilon>0$ accessing a random subset
of $(1+\epsilon)k$ encoded symbols, asymptotically suffices to recover the
$k$ input symbols with high probability.
	
Our codes have the additional benefit of logarithmic locality: a single lost
symbol can be repaired by accessing a subset of $O(\log k)$ of the remaining
encoded symbols.
This is a desired property for distributed storage  systems where symbols  are
spread over a network of storage nodes.
Beyond recovery upon loss, local reconstruction provides an efficient
alternative for reading symbols that cannot be accessed directly.
In our code, a logarithmic number of disjoint local groups is associated
with each systematic symbol, allowing \emph{multiple parallel reads}.

Our main mathematical contribution involves analyzing the rank of sparse random matrices with specific structure over finite fields.
We rely on establishing that a new family of sparse random bipartite graphs have perfect matchings with high probability. 
\end{abstract}
\section{Introduction}

Fountain codes~\cite{Byers:1998:DigitalFountain,Luby:2002:LT,
Shokrollahi:2006:RC} form a new family of linear erasure codes
with several attractive properties.
For a given set of $k$ input symbols, a Fountain code produces a potentially
limitless stream of output symbols, each created independently of others as a
random combination of input symbols according to a given distribution.
Ideally, given a randomly selected subset of $(1+\epsilon) k$ encoded symbols, a
decoder should be able to recover the original $k$ input symbols
with high probability (w.h.p.) for some small overhead $\epsilon$.
Further, Fountain codes typically emphasize on efficient
encoding and decoding algorithms.

In this work, we design a new family of Fountain codes that combine multiple
properties appealing to distributed storage.
One property that is highly desired for distributed storage codes is 
\textit{systematic form}:
the original information symbols must appear in the encoded sequence.
Their presence enables the reading of source data without decoding and is a
practical
requirement for most storage applications.
Another important property of distributed storage codes is efficient
\textit{repair}~\cite{Dimakis:NetworkCoding, XorbasVLDB, Huang:2012:azure}:
when a single encoded symbol is lost it
should be possible to reconstruct it without communicating too much information
from other encoded symbols.
A related property is that of \textit{locality} of each symbol:
the number of encoded symbols that need to be accessed to reconstruct that
particular symbol~\cite{Khan:Fast12, Oggier, Papailiopoulos:2011:simple,
Gopalan:locality}.

A key observation is that in a systematic linear code, locality is strongly connected to the sparsity of parity symbols~\cite{Gopalan:locality}, 
\textit{i.e.}, the maximum number of input symbols combined in a parity symbol.
A parity symbol along with the systematic symbols covered by it form a
\textit{local group}.
Any symbol in this group can be reconstructed via an appropriate linear
combination of the remaining member symbols.
The smaller the size of the local group, the lower the locality of the symbols in it.

In existing Fountain codes, such as LT or Raptor codes, each encoded symbol is
a linear combination of $O( \log k )$ or even constant number of input symbols on average.
However, these codes are not systematic and the low degree of the encoded symbols does not imply low locality.
Certain classes, such as Raptor codes, can be transformed
into a systematic form~\cite{Shokrollahi:2006:RC} via a preprocessing of the input.
Unfortunately, due to the additional step, parity symbols are no longer
sparse in the original input symbols.

Standard Fountain codes support computationally efficient encoding and decoding algorithms as a result of a meticulously designed encoded symbol degree distribution.
Fast decoding algorithms, however, do not translate to efficient repair:
these algorithms aim at retrieving the entire source message from a set of available symbols, and are not tailored to the needs of a single symbol reconstruction.
On the contrary, single erasures -- the most frequent scenario in a distributed storage setting -- can be efficiently repaired by an erasure code featuring low locality.
If additionally the code is in systematic form,
full scale decoding is invoked only in the unlikely event of multiple erasures that cannot be locally repaired.
In other words, the presence of the source data in the encoded sequence in conjunction with low locality renders decoding an infrequently used operation, downgrading the need for efficient decoding algorithms.

The significance of locality is not limited to the repair problem, \textit{i.e.}, the reconstruction of a symbol upon its loss.
It extends to the closely related use case of \textit{degraded reads}.
In a distributed setup each symbol is stored on a different storage node, which may be temporarily unavailable due to a variety of reasons.
Although not permanently lost, a systematic symbol may not be directly accessible, and its local group provides an efficient alternative for reading it.
The \textit{availability} of a systematic symbol naturally extends the notion of locality, measuring the number of disjoint local groups the symbol belongs to.
We define the \textit{availability} of a systematic symbol as the number of disjoint sets of encoded symbols which can be used to reconstruct that
particular symbol. 
In effect, it characterizes the number of read requests for a particular systematic symbol that can be simultaneously served.

\textbf{Our Contribution:}
We introduce a new family of Fountain codes that are systematic and also have
parity symbols with logarithmic sparsity. We show that this is impossible if
we require the code to be MDS, but is possible if we require a near-MDS property
similar to the probabilistic guarantees provided by LT and Raptor codes.

More concretely, for any $\epsilon > 0$ we construct codes that guarantee that
a random subset of $(1+\epsilon)k$ symbols suffices to recover the original $k$
input symbols w.h.p.
Our codes produce an unbounded number of output symbols, creating each parity
independently by linearly combining a logarithmic number of randomly chosen input symbols.

We show that this structure also provides logarithmic locality: each symbol in our codes is
repairable by accessing only $O(\log k)$ other coded symbols.
We further define the notion of symbol availability and show that the systematic symbols in our codes feature logarithmic availability:
with high probability, for each systematic symbol there exist $O( \log k )$ disjoint sets of symbols that can be used to reconstruct it.
This means that multiple parallel jobs can read this symbol concurrently, each by accessing one disjoint set.
This new property is motivated by the straggler performance bottlenecks observed recently in distributed storage systems~\cite{Dean:2013:tail}.


One disadvantage of our construction is higher decoding complexity.
Our codes can be decoded by solving a system of linear equations over $\mathbb{F}_q$, 
which corresponds to maximum likelihood decoding for the erasure channel.  
This naive decoding can be implemented using Gaussian elimination and
requires $O(k^3)$ steps.
Fortunately, the matrices we construct are sparse, allowing faster decoding: 
Wiedemann's algorithm~\cite{Wiedemann} can be used to decode in $O(k^2 \text{polylog} k)$ time. 
Standard Fountain codes create linear equations that can be solved just by back-substitution
which amounts to decoding complexities of $O(k \log k)$ for LT~\cite{Luby:2002:LT} and 
$O(k)$ for Raptor~\cite{Shokrollahi:2006:RC} but offer no locality. 
It remains open to construct Fountain codes that have locality and near-linear decoding complexity.


Our main technical contribution is a novel random matrix result:
we show that a family of random matrices with non independent entries have full rank with high probability.
The analysis builds on the connections of matrix determinants to flows on
random bipartite graphs, using techniques from
\cite{Erdos:RandomMatrices, Dimakis:DecentralizedErasure}.
Our key result is showing that a new family of sparse random graphs
have matchings w.h.p.
Our random graph contribution is explained in Section \ref{sec:graph-perspective}.
\section{Problem Description}
\label{sec:problem_description}

Given $k$ input symbols, elements of a finite field $\mathbb{F}_q$, we want to
encode them into $n$ symbols using a linear code.
Linear codes are described by a $k \times n$ \textit{generator} matrix
$\mathbf{G}$ over $\mathbb{F}_q$, which when multiplied by an input vector
$\mathbf{u} \in \mathbb{F}^{1 \times k}_{q}$ produces a codeword
$\mathbf{v}=\mathbf{u}\mathbf{G} \in \mathbb{F}^{1 \times
n}_q$.
Ideally, we would like $\mathbf{G}$ to have the following properties:
\begin{itemize}
 \item \textit{Systematic form}, \textit{i.e.}, a subset of the columns of
$\mathbf{G}$ forms the identity matrix, $\mathbf{I}$, which implies that the
input symbols are reproduced in the encoded sequence.
 \item \textit{Rateless property}, \textit{i.e.}, each column is created
independently. The number $n$ of columns does not have to be specified for the
encoder a priori.
Equivalently, encoded symbols can be created or removed dynamically upon
request, without recreating the entire encoded sequence.
 \item \textit{MDS property}, \textit{i.e.}, any $k$ columns of $\mathbf{G}$
have rank $k$, implying that any subset of $k$ encoded symbols suffices to
retrieve the input.
 \item \textit{Low locality}. $\mathbf{G}$ has \textit{locality} $l$ if
each column can be written as a linear combination of at most $l$ other
columns.
If the code is systematic, then sparse parities suffice to obtain good locality
\cite{Gopalan:locality}.
\item \textit{High Availability}.
A systematic symbol has \textit{availability} $t$ if it can be written as
a linear combination of $t$ disjoint sets of symbols, of cardinality $l$.
The code has availability $t^\prime$, equal to that of the least
available systematic symbol.
\end{itemize}

For any code, any sufficiently large subset of encoded symbols should allow
recovery of the original data.
The size of such a set is tightly related to the reliability of the code.
For optimal reliability, \textit{i.e.}, in the case of MDS codes, an information
theoretically minimum subset of $k$ encoded symbols suffices to decode.
When equipped with systematic form, the generator matrix of an MDS code
affords no zero coefficient in the parity generating columns.
To verify that, consider a parity column with a zero coefficient in the
$i$-th position: that parity column along with any $k-1$ systematic columns
excluding the one corresponding to the $i$-th systematic symbol form a singular
matrix.

If parities are deliberately sparse in the input symbols, seeking to improve
the code's \textit{locality}, the property that any $k$ encoded symbols
suffice to retrieve the original data has to be relaxed.
In this work, we require that for $\epsilon>0$, a set of $k^\prime=(1+\epsilon)
k$ randomly selected encoded symbols suffice to decode with high probability;
the decoder may fail, but with a probability vanishing polynomially in $k$.
We refer to a code with this property as \textit{near-MDS}.

Under this constraint, we seek codes that achieve optimal locality, which
translates into determining how sparse the parities can be without violating the
decoding guarantee.
We will show that it is impossible to recover the original message with high
probability if the parities are linear combinations of fewer than
$\Omega(\log k)$ input symbols.
Furthermore, we will design codes that achieve logarithmic sparsity in the
parities and hence, order optimal locality.
We conclude the paper investigating the availability of our
construction, showing that with high probability every systematic symbol
belongs to a logarithmic number of disjoint local groups.
\section{Prior Work}
\label{sec:prior-work}
In LT codes, the first practical realizations of Fountain codes invented by Luby
\cite{Luby:2002:LT}, the average degree of the output symbols, \textit{i.e.},
the number of input symbols combined into an output symbol, is $O\left( \log k
\right)$.
Note, however, that sparsity in this case does not imply good locality, since LT
codes lack systematic form.

Building on LT, Shokrollahi in \cite{Shokrollahi:2006:RC} introduced Raptor
codes, a different class of Fountain codes.
The core idea is to precode the input symbols prior to the application of an appropriate LT code.
By virtue of the two layer encoding, 
the per symbol encoding cost -- which corresponds to average degree of encoded symbols -- is reduced to a constant,
while the $k$ input symbols can be retrieved in linear time by a set of $(1+\epsilon)k$ encoded symbols, with probability of failure at most inversely polynomial in $k$.
However, the original Raptor design does not feature the highly desirable systematic form.
Further, similar to the LT codes, the constant average degree of the encoded symbols does not imply good locality.

In the same work~\cite{Shokrollahi:2006:RC},
Shokrollahi provided a construction that yields a systematic flavor of Raptor codes.
The Raptor encoding is not applied directly on the input symbols, rather on the output of a preprocessing step of complexity $O(k^2)$.
The source symbols appear in the encoded stream, but due to the preprocessing step the parity symbols are no longer sparse in the original input symbols, despite their constant average degree.



Gummadi in his thesis \cite{Gummadi:thesis} was the first to consider the design of Fountain
codes explicitly oriented for storage applications, \textit{i.e.}, codes that feature systematic form and efficient repair.
The latter is quantified by \textit{repair complexity}: the average number of symbol operations performed to repair a set of erased symbols.
Gummadi proposes systematic variants of LT and Raptor codes 
that feature low (even constant) expected repair complexity.
However, the overhead $\epsilon$ required for decoding is suboptimal: it cannot be made arbitrarily small.

Our main result is the analysis of the rank of a new family of sparse random matrices over  $\mathbb{F}_q$.
In particular, we investigate the probability that a $k \times (1+\epsilon) k$ matrix comprising any number $0 \le s \le k$ of systematic columns and $(1+\epsilon) k -s $ random $O(\log k)$-sparse columns has full rank.
There is a long line of work on the distribution of the rank of sparse random matrices over a finite field (\textit{e.g.}, work by Karp~\cite{Karp}, Kovalenko~\cite{kovalenko}, Balakin~\cite{Balakin}, Cooper~\cite{Cooper} and references therein).
That line of work, however, typically focuses on random matrices whose entries are independently distributed.
In our case, systematic columns carry exactly one nonzero entry, while the number of nonzero entries in the remaining columns is strictly upper bounded, rendering  column entries dependent.


\section{Repairable Fountain Codes}
\label{sec:our-construction}
We introduce a new family of Fountain codes that are systematic and also have
sparse parities.
Each parity symbol is a random linear combination of up to $d$ randomly
chosen input symbols.
Due to their randomized nature, our codes provide a probabilistic guarantee on
successful decoding.
In particular, we require that a set of $k^\prime=(1+\epsilon)k$ randomly
selected encoded symbols, for arbitrarily small $\epsilon > 0$ can be decoded
successfully with high probability, \textit{i.e.}, with
probability of failure vanishing like $1/\text{poly}(k)$.
We show that under this constraint, $d$ must be of at least logarithmic
order in $k$, \textit{i.e.}, $d = \Omega\left( \log{k}\right)$.
Surprisingly, however,  a logarithmic sparsity level for the parity symbols
is also achievable, hence  $d = \Theta\left( \log{k}\right)$.
The sparsity of the parity columns corresponds to the \textit{locality} of the
code family.
Our main result, which is asymptotic in $k$, is established in
Theorem \ref{thm:main-theorem}, at the end of this section.
We conclude the section with a study of the \textit{availability} of our
construction.

Given a vector $\mathbf{u}$ of $k$ input symbols in $\mathbb{F}_q$, the code is
a linear mapping of $\mathbf{u}$ to a vector $\mathbf{v}$ of higher dimension
$n$ through a $k \times n$ matrix $\mathbf{G}$.
The encoded sequence comprises an un-encoded copy of the $k$ input
symbols augmented by parity symbols, hence the \textit{systematic} form.
Without loss of generality, we may assume that $\mathbf{u}$ lies in the
first indices of $\mathbf{v}$ followed by the parity symbols.
A single parity symbol is constructed in a two step process.
First, $d$ input symbols are successively selected uniformly at random,
independently, with replacement.
Then, a coefficient is uniformly drawn from $\mathbb{F}_q$ for each
symbol previously selected.
The parity is the linear combination of the symbols selected in the first step,
weighted with the coefficients drawn in the second step.
The same procedure is repeated independently for subsequent parity symbols.
The independent construction of parities is the hallmark of a Fountain code.
\begin{figure}[htbp!]
  \centering
  \begin{tikzpicture}[scale=1,->,>=stealth',shorten >=1pt,auto,node distance=1cm, thick,
  every node/.style={scale=1},
symbol node/.style={circle, draw,minimum height=0.7cm},
parity node/.style={rectangle, draw,minimum height=0.7cm,minimum width=0.7cm}
]

    \node[symbol node] (U1) at (0,0) {$u_1$};
    \node 		(U2) [below of=U1] {$\vdots$};
    \node[symbol node] (U3) [below of=U2] {$u_i$};
    \node		(U4) [below of=U3] {$\vdots$};
    \node[symbol node] (U5) [below of=U4] {$u_k$};

    \node[symbol node] (V1) at (3, 1) {$v_1$};
    \node 		(V2) [below of=V1] {$\vdots$};
    \node[symbol node] (V3) [below of=V2] {$v_i$};
    \node		(V4) [below of=V3] {$\vdots$};
    \node[parity node] (V5) [below of=V4] {$v_j$};
    \node		(V6) [below of=V5] {$\vdots$};
    \node[parity node] (V7) [below of=V6] {$v_n$};

    \draw 		(V1) edge [bend right, out=-20, in=-160] (U1);
    \draw 		(V3) edge [bend right, out=-20, in=-160] (U3);

    \draw 		(V5) edge [out=130, in=0] node[above, near start]
{$w_{ij}$} (U3) ;
    \draw[color=gray] (V5) edge [out=140, in=22.5] (1.5, -2.5);
    \draw[color=gray] (V5) edge [out=150, in=45] (1.5, -3);

    \draw (V7) edge [out=160, in=-45] (1.5, -4);
    \draw (V7) edge [out=210, in=315, looseness=1] (1.5, -5);
    \node at (1.5, -4.5) {$\vdots$};
    \node[above, rotate=90] at (1.5, -4.5) {$d(k)$};

%
%

  \end{tikzpicture}
  \caption{Bipartite graph $G=(U,V,E)$ corresponding to our randomized code
construction: $U$ is the set of $k$ input symbols and $V$ is the set of $n$
encoded symbols.
The first $k$ vertices in $V$ have degree one and correspond to the systematic
part of the encoded sequence.
Each one of the remaining vertices independently and uniformly throws $d(k)$
edges,  with two or more edges possibly landing on the same vertices in $U$.}
  \label{fig:bipartite_near_mds}
\end{figure}
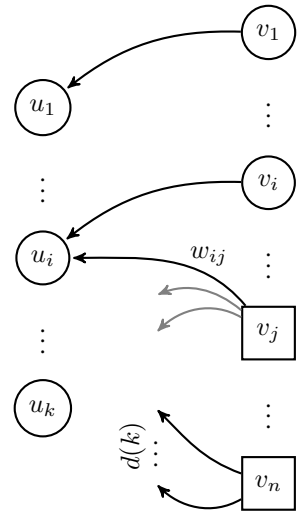
 
It is useful to describe our randomized construction through a correspondence
to a bipartite graph $G=(U,V,E)$, depicted in
Figure \ref{fig:bipartite_near_mds}.
The set $U$ of vertices on the left side corresponds to the $k$ input symbols,
and the set $V$ on the right corresponds to the $n$ symbols of the encoded
sequence.
An edge $(u_i,v_j) \in E$ if the input symbol $u_i \in U$ is one of the symbols
participating in the formation of the encoded symbol $v_j \in V$.
Each of the $k$ first vertices in $V$ has degree equal to one and is connected
to a distinct input symbol.
These $k$ vertices correspond to the deterministically constructed systematic
part of the encoded sequence.
Each one of the remaining vertices corresponds to a parity symbol and
forms its neighborhood through the following randomized procedure.
Node $v_j$ throws an edge to a vertex in $U$ selected uniformly at random.
This step is repeated a total number of $d$ times, independently.
At the end of this process, vertex $v_j$ has selected a subset
$\mathcal{N}(v_j)$
of vertices in $U$, the neighborhood of $v_j$.

The parity symbol corresponding to vertex $v_j$  is a random linear combination
of the input symbols in its neighborhood.
Slightly abusing notation, $v_j$ is used to denote both the vertex and
the corresponding entry in the encoded sequence $\mathbf{v}$.
The $j$-th encoded symbol can be written as
\begin{equation}
 v_j = \sum_{i \in \mathcal{N}(v_j)} w_{ij}u_i,
\end{equation}
where $w_{ij}$'s are randomly selected uniformly and independently from
$\mathbb{F}_q$.
The coefficients $w_{ij}$'s can be embedded in the graph representation as
weights on the corresponding edges.
The edges of the systematic part have unit weights.

The degree of $v_j \in V$, \textit{i.e.}, the size of its neighborhood
$|\mathcal{N}(v_j)|$ can be at most equal to $d$, the number of edges thrown.
It can be strictly smaller if a vertex in $U$ is selected multiple times.
However, when $d$ is much smaller than $k$, $|\mathcal{N}(v_j)|$
will be equal to $d$ with high probability, and a parity symbol will be a
linear combination of $d$ input symbols.
To emphasize that $d$ is allowed to grow as a function of $k$,
we will hereafter denote it by $d(k)$.

Returning to the matrix representation, the code construction corresponds to a
family of generator matrices $\mathbf{G}$ of the form $\mathbf{G}=
\left[ \; \mathbf{I}_{k \times k}\;|\; \mathbf{P}\;\right]$.
Every encoded symbol corresponds to a column of $\mathbf{G}$.
The identity part confers the systematic form.
$\mathbf{P}$, the part responsible for the construction of the parity
symbols, is a random matrix whose columns are sparse, each bearing at most
$d(k)$ nonzero entries.
Any $k$ encoded symbols corresponding to linearly independent columns of
$\mathbf{G}$ suffice to retrieve the input $\mathbf{u}$.
Conversely, reconstructing $\mathbf{u}$ from a randomly chosen set of $k^\prime
> k$ encoded symbols is possible only if $k$ symbols among them correspond to
linearly independent columns.
Therefore, the key property required for successful decoding
of a set $\mathcal{S}$ of $k^\prime$ randomly selected encoded symbols is
that $\mathbf{G}_{\mathcal{S}}$, the $k \times k^\prime$ matrix formed by the
corresponding columns of $\mathbf{G}$, including
any combination of systematic and parity parts, has full rank w.h.p.

The probability that the input can be recovered from $k^\prime$ randomly
selected encoded symbols increases with $d(k)$.
Equivalently, fewer encoded symbols suffice to attain a certain probability
of successful decoding.
We have highlighted the extreme case of systematic MDS codes: for the
optimal guarantee that any $k$ symbols suffice to recover the input, $d(k)$ can
be no less than $k$.
To gain further insight, note that for a set of $k^\prime$ encoded symbols to
be successfully decoded, it is necessary that all input symbols are covered by
that set.
As $d(k)$ decreases, so does the probability that a particular symbol is
covered by the parities in a set of $k^\prime$ encoded symbols, impacting the
decoding guarantees.

On the other hand, as noted Section \ref{sec:problem_description}, a
systematic code with sparse parities has good locality.
The relation between $d(k)$ and locality is straightforwardly quantifiable:
any parity symbol $v$ is a linear combination of at most $d(k)$ systematic
symbols.
Also rearranging the terms, any systematic symbol $u$ covered by a parity $v$
can be written as
a linear combination of $v$ and the remaining systematic symbols covered by $v$.
Under the assumption that there exist at least one parity symbol covering every
systematic 
symbol $u$, the code has locality at most $d(k)$.

In summary, decreasing $d(k)$ improves the locality of the code, with a toll
on the probability of successful decoding of a random set of
$k^\prime = (1+\epsilon)k$ encoded symbols, where $\epsilon$ is a positive constant denoting the decoding overhead.
Our primary contribution, portrayed in Theorem \ref{thm:main-theorem}, is
identifying how \textit{small} $d(k)$ can be to ensure that a randomly selected
set of $k^\prime = (1+\epsilon)k$ symbols is decodable, or equivalently that a $k \times
k^\prime$ submatrix $\mathbf{G}_{\mathcal{S}}$ of $\mathbf{G}$ is full rank, 
with high probability.
\begin{theorem}
  \label{thm:main-theorem}
  Consider a matrix $\mathbf{G}=\left[ \;\mathbf{I}_{k \times k}\;| \;
\mathbf{P} \;
  \right]$, where each column of $\mathbf{P}$ is independently constructed as follows:
  $(i)$ $d(k) = c\cdot \log{k}$ out of the $k$ entries are selected uniformly at random with replacement, 
  and $(ii)$ a value drawn uniformly at random over $\mathbb{F}_q$ is independently assigned to each entry selected in step $(i)$.
  Then, for constant $c  = (8 + \rho +
2\epsilon)/\epsilon > 0$, and $q > k$,
  a randomly selected $k\times (1+\epsilon)k$ submatrix
  $\mathbf{G}_{\mathcal{S}}$ of $\mathbf{G}$ containing any number of systematic columns is full rank with probability at
least
  $1-(k/q)-k^{-\rho}$.
\end{theorem}
Theorem \ref{thm:main-theorem:converse} establishes a converse
result stating that the \textit{sufficient} value of $d(k)$ of Theorem
\ref{thm:main-theorem} is order-optimal for our construction.
\begin{theorem}
 \label{thm:main-theorem:converse}
  (Converse) If each column of $\mathbf{P}$ is generated independently as
described with at most $d(k)$ nonzero entries, then $d(k) = \Omega(\log{(k)})$ is
necessary for a random $k\times k^\prime$ submatrix $\mathbf{G}_{\mathcal{S}}$
of $\mathbf{G}=\left[ \;\mathbf{I}_{k \times k}\;| \; \mathbf{P} \; \right]$ to
be full rank
w.h.p.
\end{theorem}

From the two theorems, it follows that our codes achieve optimal locality with a
logarithmic degree for every parity symbol.
Original data is reconstructed in $O(k^3)$ using Maximum Likelihood (ML)
decoding, which corresponds to solving a linear system of $k^\prime$ equations
over $\mathbb{F}_q$.
Note, however, that the Wiedemann algorithm \cite{Wiedemann}
can reduce complexity to $O\left(k^2\text{polylog}(k)\right)$ on average, exploiting the
sparsity of the linear equations, with negligible extra memory requirement.
Finally, we note a drawback of our analysis: in order to achieve vanishingly
small probability of failure as $k$ grows, the size of the field must grow
accordingly.
It suffices, however, that the number of bits per symbol grows logarithmically in $k$:
a symbol size of $(t+1)\log k$ bits, $t>0$, implies that $k/q = 1/k^t$.

\begin{figure}[htbp!]
 \centering
    \begin{tikzpicture}[scale=0.67,
    every node/.style={scale=0.67},->,>=stealth',shorten >=1pt,auto,node
distance=1.5cm, thick,
    symbol node/.style={circle, draw, minimum height=1cm},
    parity node/.style={rectangle, draw, minimum height=1cm, minimum
width=1cm}
    ]

    \node[symbol node] (U1) at (0,0) {$u_1$};

    \node[parity node] (V1) at (4, 3) {$v_1$};
    \node[parity node] (V2) [below of=V1] {$v_2$};
    \node[parity node] (V3) [below of=V2] {$v_3$};
    \node              (V4) [below of=V3] {$\vdots$};
    \node[parity node] (V5) [below of=V4] {$v_M$};

    \draw[thin] (V1) edge [looseness=.2, out=180, in=45] (U1);
    \draw[thin] (V2) edge [looseness=.2, out=180, in=22.5] (U1);
    \draw[thin] (V3) edge [looseness=.2, out=180, in=0] (U1);
    \draw[thin] (V5) edge [looseness=.2, out=180, in=-45] (U1);
    \node[symbol node] (U2) at (8,5.15)   {$u_2$};
    \node[symbol node] (U3) [below of=U2] {$u_3$};
    \node[symbol node] (U4) [below of=U3] {$u_4$};
    \node[symbol node] (U5) [below of=U4] {$u_5$};
    \node[symbol node] (U6) [below of=U5] {$u_6$};
    \node              (U7) [below of=U6] {$\vdots$};
    \node[symbol node] (U8) [below of=U7] {};
    \node[symbol node] (U9) [below of=U8] {$u_k$};
    \node[rectangle, fit = (U2) (U3) (U4), draw=gray, thick, rounded  corners =
5, inner sep = 0.15cm] (footprintV1){};
    \node[rectangle, fit = (U3) (U4) (U5), draw=gray, rounded corners =
5, inner sep = 0.05cm] (footprintV2){};
    \node[rectangle, fit = (U5) (U6), draw=gray, thick, rounded  corners =
5, inner sep = 0.15cm] (footprintV3){};
    \node[rectangle, fit = (U8) (U9), draw=gray, thick, rounded  corners =
5, inner sep = 0.15cm] (footprintVM) {};
    \draw [->, black, thin, looseness=.3, out=40, in=200] (V1) edge
($(U2)+(-0.6, 0.55)$);
    \draw [->, black, thin, looseness=.3, out=-20, in=170] (V1) edge
($(U4)+(-0.6, -0.55)$);
    \draw [->, black, thin, looseness=.3, out=45, in=200] (V2) edge
($(U3)+(-0.55, 0.4)$);
    \draw [->, black, thin, looseness=.3, out=-30, in=160] (V2) edge
($(U5)+(-0.55, -0.4)$);
    \draw [->, black, thin, looseness=.3, out=30, in=180] (V3) edge
($(U5)+(-0.6, 0.55)$);
    \draw [->, black, thin, looseness=.3, out=-30, in=180] (V3) edge
($(U6)+(-0.6, -0.55)$);
    \draw [->, black, thin, looseness=.3, out=20, in=160] (V5) edge
($(U8)+(-0.6, 0.55)$);
    \draw [->, black, thin, looseness=.3, out=-45, in=160] (V5) edge
($(U9)+(-0.6, -0.55)$);
    \draw[-, line width=.5pt,black,decorate,decoration={amplitude=7pt,brace}]
  ($(footprintV1.north east)+(0.2, 0)$) -- ($(footprintV1.south east)+(0.2,
0)$);
   \node[anchor=west, rotate=0, text width=2.7cm, font=\large] at
($(U3)+(1.5, 0)$) {$\mathcal{F}_{v_1}$: The \textit{footprint} of $v_1 \in
\mathcal{P}_{u_1}$.};
    \node (disjointfooprint) at ($(U7)+(1.1,0)$)  {};
    \draw[ ->, line width=.3pt, black, dashed, looseness=0.5, out=90, in=290]
  (disjointfooprint) edge ($(footprintV3.south east)+(0.2, 0)$);
    \draw[ ->, line width=.3pt, black, dashed, looseness=0.5, out=270, in=70]
  (disjointfooprint) edge ($(footprintVM.north east)+(0.2, 0)$);
   \node[anchor=west, rotate=0, text width=2.7cm, font=\large] at
($(U7)+(1.5, 0)$) {$\mathcal{F}_{v_3}$ and $\mathcal{F}_{v_M}$ are disjoint. $v_3$
and $v_M$ are \textit{isolated}.};
  \node (availabilityset) at ($(U1)+(1, -4)$) {};
  \draw[ ->, line width=.3pt, black, dashed, bend right]
  (V1) edge (availabilityset);
  \draw[ ->, line width=.3pt, black, dashed, bend right]
  (V3) edge (availabilityset);
  \draw[ ->, line width=.3pt, black, dashed, bend right]
  (V5) edge (availabilityset);
  \node[anchor=north, rotate=0, text width=4.5cm, font=\large] at
($(availabilityset)+(1, 0)$) {$v_1, v_3$ and $v_M$ are pairwise
\textit{isolated}. The \textit{availability} of $u_1$ is at least $3$.};
\end{tikzpicture}
\caption{%
  The systematic symbol $u_1$ is covered by $M$ parities, $v_1,\hdots,
  v_M$.
  The footprint of such a parity $v_i$ (with respect to $u_1$) is the set of
  symbols it covers excluding $u_1$, \textit{e.g.}, the footprint of $v_1$ is
  $\mathcal{F}_{v_1} = \lbrace u_2, u_3, u_4\rbrace$.
  Symbol $u_1$ can be reconstructed using any of the above parity symbols and
  its footprint.
  In this example, $u_1$ has availability at least $3$ since the footprints of
  parities $v_1$, $v_3$, and $v_M$ are disjoint.
}
\label{fig:footprints}
\end{figure}
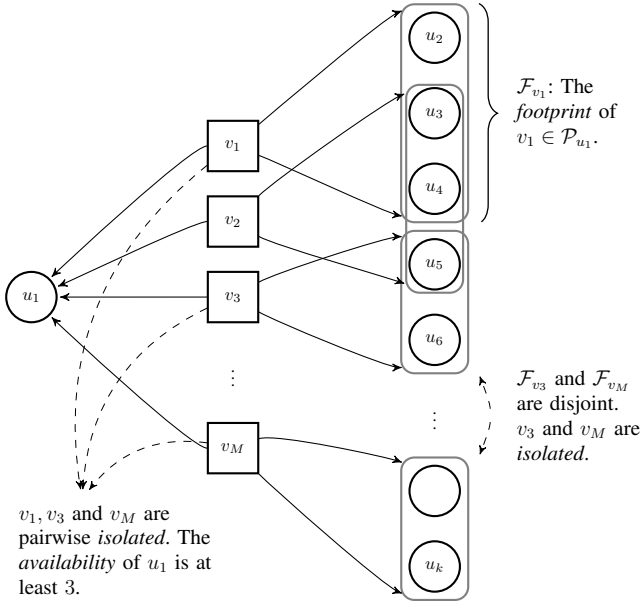
Thus far, we have seen that our randomized construction achieves logarithmic
locality:
every encoded symbol belongs to at least one \textit{local group} of
cardinality $d(k)+1$.
There is a one-to-one correspondence between the local groups $u$ belongs to
and the parity symbols that cover $u$.
Let $\mathcal{P}_u = \lbrace v_j : k+1 \le j \le n, u \in
\mathcal{N}(v_j)\rbrace$
be the subset of parities that cover $u$.
Its cardinality is a binomial random variable since every generated parity
independently covers $u$ with some probability.
If the total number of parities generated is $rk$ for some constant $r > 0$,
\textit{i.e.}, proportional
to the length of the input, then for $d(k) = c\log(k)$ of Theorem
\ref{thm:main-theorem} every systematic symbol is in fact covered by a
logarithmic number of parities w.h.p.
\begin{theorem}
  \label{thm:ub-prob-exist-u-few-parities}
  Let $rk$ be the total number of parities generated, for a constant $r>0$,
  with each parity symbol constructed as a linear combination of $d(k) = c
\log(k)$ independently selected symbols uniformly at random with replacement.
  The expected number of parities covering a systematic symbol $u$ is
  \begin{align}
  rc\log(k) - \frac{rc\log^2(k)}{k} \le E \left[ |\mathcal{P}_u| \right]  \le
  rc\log(k).
  \end{align}
  Further,
  \begin{align}
  \text{Pr}\left( \exists u : \left|\mathcal{P}_u\right| \right. &\le \left.
  (1-\epsilon)E[\left|\mathcal{P}_u\right|] \right)
  \nonumber\\
  &\le \frac{1}{k^{\frac{rc\epsilon^2}{2}-1}} \exp\left(
  \frac{\epsilon^2}{2} \frac{rc\log^2(k)}{k} \right).
  \label{eq:ub-prob-exist-u-few-parities}
  \end{align}
\end{theorem}
For any $\epsilon$, an appropriate choice of $r$, and $c$ achieves a vanishing
bound in \eqref{eq:ub-prob-exist-u-few-parities}.
The above result states that with high probability all input symbols are
covered by at least $(1-\epsilon)rc\log(k)$ symbols for some $\epsilon>0$.
In the following, we will omit the constant $1-\epsilon$ for simplicity, and
assume that every systematic symbol is covered by at least $m = rc\log(k)$.

The \textit{availability} of the input symbol $u$ is the cardinality
of the largest subset of \textit{local groups} containing $u$ whose only
common element is $u$.
More formally, consider a parity $v_i$ that covers the systematic symbol $u$,
\textit{i.e.}, $v_i \in \mathcal{P}_u$.
Then $v_i$ is a linear combination of the symbols
in $\{u\} \cup \mathcal{F}_{v_i}$, where $\mathcal{F}_{v_i}$ contains the remaining
symbols covered by $v_i$.
The set $\mathcal{F}_{v_i}$ is referred to as the \textit{footprint} of $v_i$ with
respect to $u$.
Note that the \textit{footprint} of a parity symbol $v_i$ is a concept relative
to the systematic symbol $u$ under consideration.
Two parities $v_i, v_j \in \mathcal{P}_u$ are \textit{isolated} if their
footprints are disjoint.
An example is depicted in Figure \ref{fig:footprints}.
The cardinality of the largest subset of $\mathcal{P}_u$ such that parities
are pairwise isolated corresponds to the availability of the symbol $u$.

Under the assumption that every systematic symbol is covered by at least a
logarithmic number of parities, Theorem \ref{thm:availability-of-the-code}
states that every systematic symbol has a logarithmic availability with high
probability.
\begin{theorem}
  \label{thm:availability-of-the-code}
  Assuming that every input symbol $u$ is covered by at least $m = rc\log
  k$ parity symbols created independently as described in Section
\ref{sec:our-construction}, for sufficiently large $k$, $\lambda>0$ and
$\alpha > 1 -
  \sqrt{\lambda/rc} - \frac{\log^4 k}{k}$,
  \begin{align}
    \text{Pr}\left( \exists u \text{ not $\left(\alpha m \right)
$-available}\right)
    \le \frac{1}{k^\lambda}.
  \end{align}
\end{theorem}

\section{A graph perspective}
\label{sec:graph-perspective}
The randomized construction of our erasure code is naturally mapped to a
family of random bipartite graphs $G = (U,V,E)$ depicted in Figure
\ref{fig:bipartite_near_mds}.
The correspondence, established early in Section
\ref{sec:our-construction}, lays the foundation for all subsequent
analysis, but also provides an alternative viewpoint for our results as
purely structural properties of the random graphs, setting the coding
background aside.

Under the graph perspective, Theorem \ref{thm:main-theorem} states that a
randomly chosen subgraph of $G$ has a perfect matching.
First, consider a balanced random bipartite graph where $|U|=|V|=k$ and each vertex of $V$ is randomly connected to $d(k) = c \log k$ nodes in $U$. 
A classical result by Erd\H{o}s and Renyi~\cite{Erdos:RandomMatrices} shows that these graphs will have perfect matchings with high probability.
However, the graphs we consider are unbalanced, with $|U|=k$ and $|V|=k^\prime = (1+\epsilon)k$ vertices, for $\epsilon \ge 0$, like the one depicted in Figure \ref{fig:bipartite-graph-random-unbalanced-subgraph}.
Out of the $k^\prime$ vertices in $V$, $s$ vertices are special with degree $1$, corresponding to systematic symbols, and $k^\prime -s$ are connected to  $d(k) = c \log k$ vertices in $U$.
In that sense, if we set $s=0$ and $\epsilon=0$ we recover the classical result of~\cite{Erdos:RandomMatrices}.
Our additional analysis is required because our proof needs to hold for all values of $s$ ranging from $0$ up to $k-1$.
This corresponds to $s$ vertices in $U$ being trivially matched with those vertices in $V$ that have degree $1$, and the remaining $k-s$ vertices being matched via the random edges.
More formally, let $G_\mathcal{S} = (U, V_\mathcal{S}, E_\mathcal{S})$ be a
subgraph of $G$, where $V_\mathcal{S} \subseteq V$ is a subset of $k^\prime$
vertices, and $E_\mathcal{S} \subseteq E$ is the subset of edges incident to
$V_\mathcal{S}$.
Theorem \ref{thm:main-theorem} states that $G_\mathcal{S}$ has a perfect
matching, which for the unbalanced bipartite graph is a matching that
saturates all $k$ vertices in $U$.
In fact, this observation is a key component in the proof of Theorem
\ref{thm:main-theorem}.
The transition from the perfect matching of a subgraph to the rank of
a submatrix which finalizes the proof requires only that the random
coefficients are drawn from a large enough field.
Along the same lines, Theorem \ref{thm:main-theorem:converse} states that
randomly throwing $d(k) = \Omega(\log k)$ edges on the parity symbols are
necessary to guarantee that a vertex in $U$ is connected in $G_\mathcal{S}$
with high probability.

The vertices in $V$, the right hand side of $G$, have by construction degree
either equal to one or approximately equal to $d(k) = c \log(k)$ for some
$c>0$.
Theorem \ref{thm:ub-prob-exist-u-few-parities} states
that when $|V|$ increases linearly in $|U|$, the degree of the vertices in $U$
is concentrated around its expectation, which is proportional to $d(k)$.

\begin{figure}[htbp!]
 \centering
    \begin{tikzpicture}[scale=0.8,
    every node/.style={scale=0.8},->,>=stealth',shorten >=1pt,auto,node
distance=1.1cm, thick,
    symbol node/.style={circle, draw, minimum height=0.7cm},
    parity node/.style={rectangle, draw, minimum height=0.7cm, minimum
width=0.7cm},
    extra parity node/.style={rectangle, draw, densely dotted, minimum
height=0.7cm, minimum width=0.7cm}
  ]
    \node[parity node] (V1) at (0, 0) {$v_1$};
    \node[parity node] (V2) at (2, -1) {$v_2$};
    \node[parity node] (V3) at (3, 1.5) {$v_3$};
    \node[extra parity node] (extra1) at (5, -0.5) {$+$};
    \node[parity node] (V5) at (3.3, -2.5) {$v_M$};
    \node[extra parity node] (extra2) at (0.5, -2) {$+$};

    \draw [-] (V1) -- (V2);
    \draw [-] (V3) -- (V2);
    \draw [-, densely dashed] (V1) -- (extra2);
    \draw [-, densely dashed] (V3) -- (extra1);

    \path [name path=teleso, fill=gray!10,opacity=0.20,line
  width=1,draw=black] plot [smooth cycle] coordinates
  { ($(V1)+(0.6, -0.6)$)
  ($(V1)+(-0.5, -0.5)$)
  ($(V1)+(-0.5,0.5)$)
  ($(V3)+(0.4, 0.7)$)
  ($(V5)+(1,2.5)$)
  ($(V5)+(0.6,-0.5)$)
  ($(V5)+(-0.7,-0.2)$)
  ($(V2)+(0.6,1)$)};
  \end{tikzpicture}
  \caption{
  Graph $H_u$ with vertices corresponding to the parities covering $u$.
  Two parities are connected if their footprints overlap.
  This example is in accordance with that of Figure \ref{fig:footprints}.
  }
  \label{figure:parity-connectivity-graph}
\end{figure}
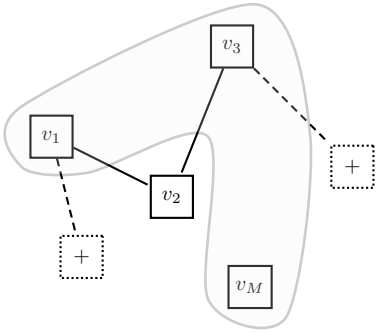
Finally, the availability of a symbol $u$ can be mapped to the independence
number of a random graph $H_u = (\mathcal{P}_u, E_u)$.
$\mathcal{P}_u$ denotes the set of parities covering $u$, or equivalently the
vertices of $V$ in the neighborhood of $u$.
The set of edges $E_u$ is constructed as follows:
for $v_i, v_j \in \mathcal{P}_u$, $(v_i, v_j) \in E_u$ if and only if the
footprints of $v_i$ and $v_j$ overlap.
The availability of $u$ is the cardinality of the maximum independent set in
$H_u$, which is shown to be at least a constant fraction of $|\mathcal{P}_u|$,
 and hence logarithmic in $k$, with high probability.
Theorem \ref{thm:availability-of-the-code} states that this
property holds simultaneously for all $u \in U$.

\section{Simulations}
\begin{figure}[htbp!]
   \centering
   \includegraphics[width=0.49\textwidth]{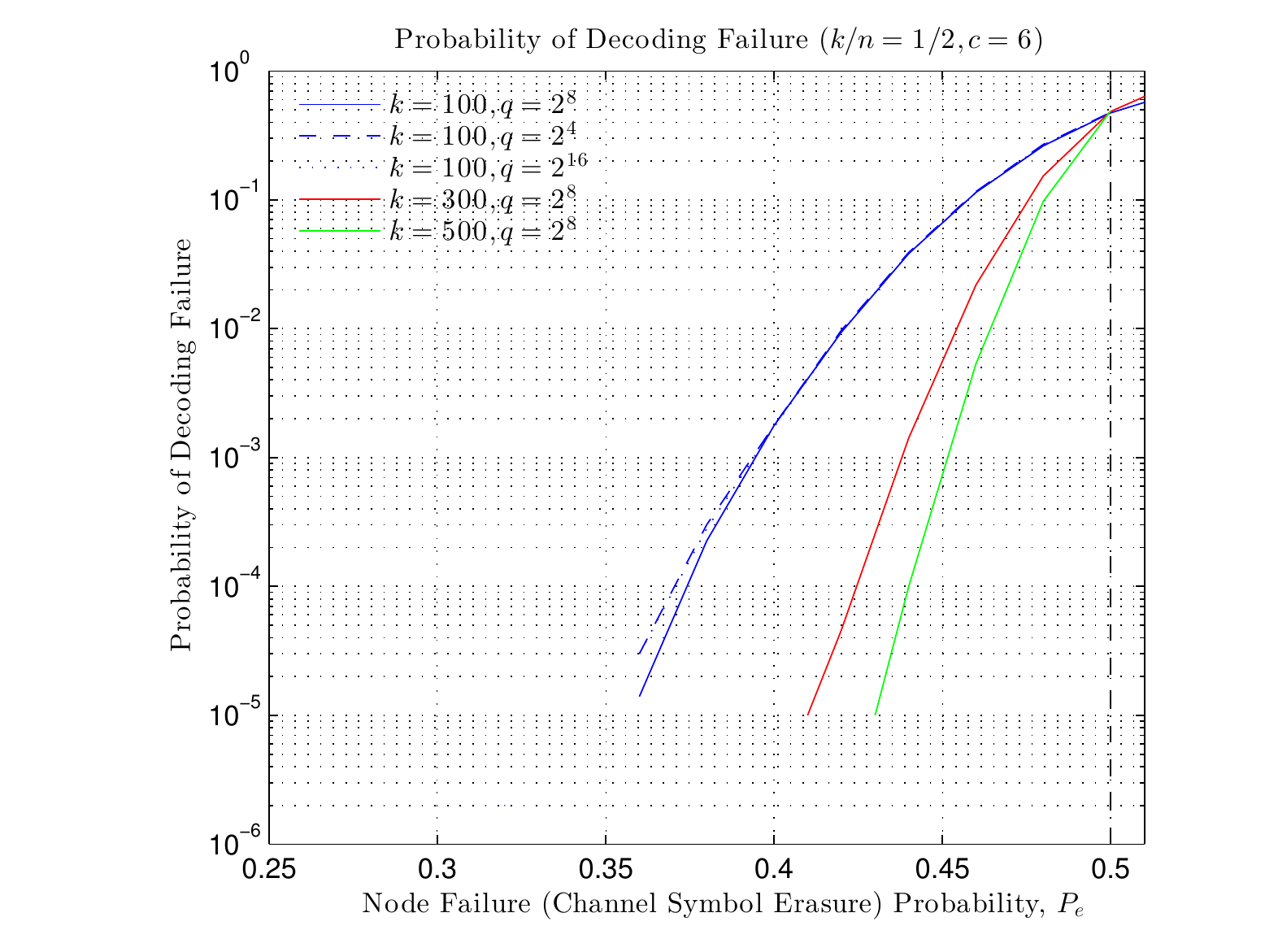}
   \caption{
      Probability of decoding failure versus the probability of symbol erasure, $P_e$.
      The probability is estimated over multiple $10^3$ randomly generated code instances with $k$ input symbols and rate $1/2$, and $10^3$ independent trials per instance .
      The degree of parity symbols is equal to $d(k) = \lceil c\log(k) \rceil$, with $c = 6$.
      A fixed value of $P_e$, corresponds to an expected decoding overhead $\overline{\epsilon} = 1 - 2P_e$.
   }
   \label{fig:simulations-B}
\end{figure}
In this section we experimentally evaluate the probability that decoding fails when a randomly selected subset of encoded symbols is available at the decoder.
Since our codes are rateless, we can set any target desired rate and examine the performance under random erasures.
In this experiment we set the rate equal to $1/2$; the generator matrix comprises the $k$ columns of the identity matrix and $k$ parity generating columns, constructed randomly and independently as described in Section \ref{sec:our-construction}.
The degree of the parities is upper bounded by $d(k) = \lceil c\log(k) \rceil$, where the pre-log factor is arbitrarily set to a small constant value.
Decoding fails exactly when the columns  corresponding to the encoded symbols available to the decoder form a matrix whose rank is strictly less than $k$.

A first series of experiments considers a sequence of random trials in which individual encoded symbol are erased in independently with probability $P_e$.
The ensemble of surviving symbols is available to the decoder.
This corresponds to the transmission through an erasure channel with erasure probability $P_e$.
The cardinality of the decoding set is a binomial random variable with expected value equal to  $(1-P_e)2k$, which amounts to an expected decoding overhead $\overline{\epsilon} = 1 - 2 P_e$.
A total of $10^3$ code instances are generated and each is subjected to  $10^3$ trials per value of $P_e$.
Fig.~\ref{fig:simulations-B} depicts the probability of decoding failure versus the channel erasure probability, $P_e$.
The experiment is repeated for three values of $k$: $k=100$, $300$, and $500$.
The field size is set to $q=2^8$, \textit{i.e.}, a single byte per symbol, for all values of $k$.

In a second series of experiments, the decoding set of cardinality equal to $k^\prime = \lceil (1+\epsilon)k \rceil$ is selected uniformly at random in each trial from the set of $2k$ encoded symbols.
Fig.~\ref{fig:simulations-B2} depicts the estimated probability of decoding failure versus the decoding overhead $\epsilon$.
\begin{figure}[htbp!]
   \centering
   \includegraphics[width=0.49\textwidth]{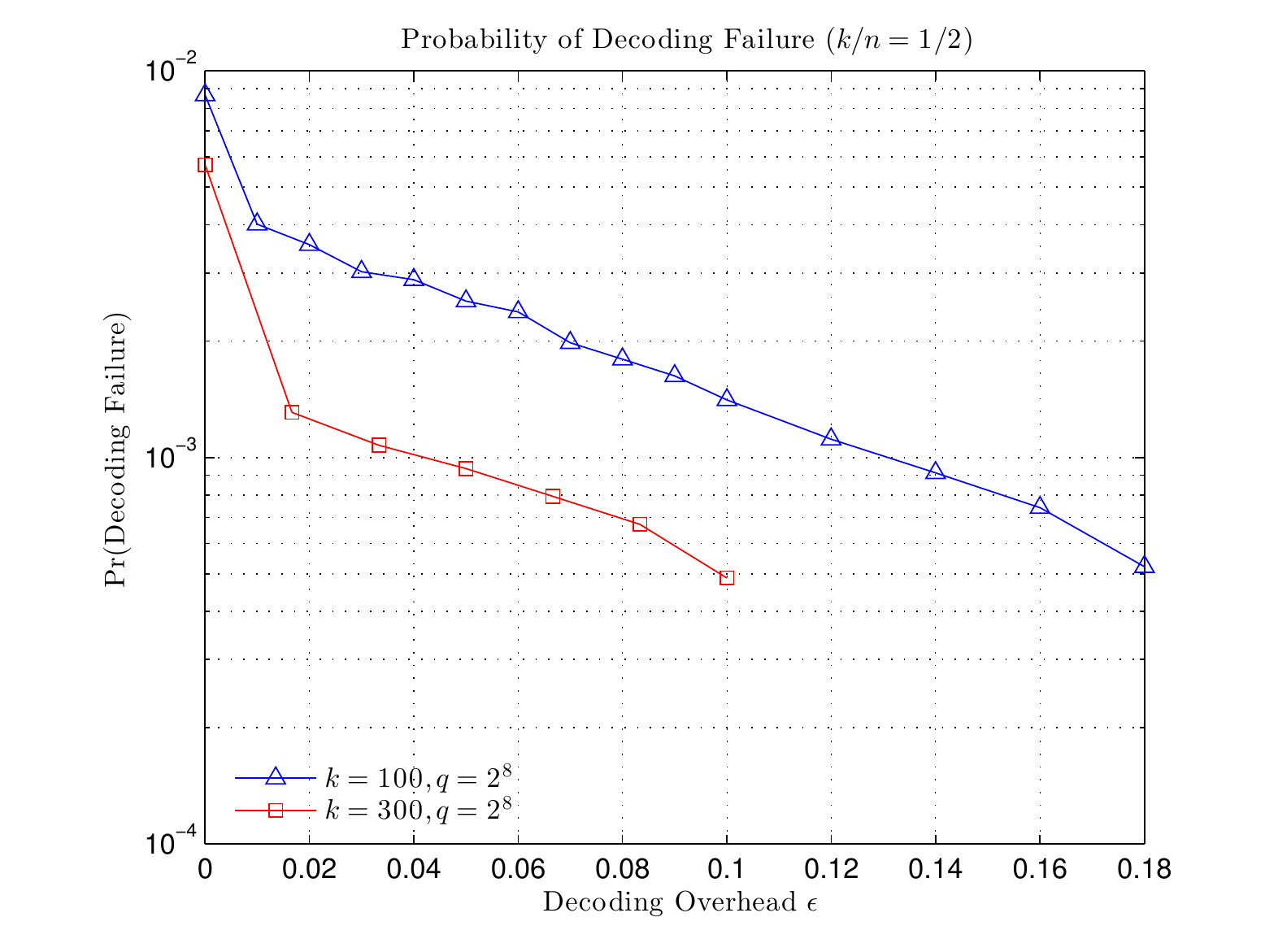}
   \caption{
      Probability of decoding failure versus the decoding overhead $\epsilon$.
      The probability is evaluated over multiple $10^3$ randomly generated code instances with $k$ input symbols and rate $1/2$, and $10^3$ independent trials per instance.
      The degree of parity symbols is equal to $d(k) = \lceil c\log(k) \rceil$, with $c = 4$.
   }
   \label{fig:simulations-B2}
\end{figure}

\section{Analysis and Proofs}
\subsection{Proof of Theorem \ref{thm:main-theorem}}
Theorem \ref{thm:main-theorem} states that when $\mathbf{G}$ is constructed as
described in section \ref{sec:our-construction}, a randomly selected
$k\times k^\prime$ submatrix $\mathbf{G}_{\mathcal{S}}$ is full rank w.h.p.
Equivalently, there exists a set of indices $\mathcal{K} \subset \mathcal{S}$
with $|\mathcal{K}|=k$ such that $k \times k$ submatrix
$\mathbf{G}_\mathcal{K}$ is nonsingular.
More formally,
\begin{equation}
 \text{Pr}\left( \exists\;
 \mathcal{K} \subseteq \mathcal{S}
:\text{det}(\mathbf{G}_\mathcal{K}) \neq 0 \right)
=1-\left(\frac{k}{q}+o(1)\right).
\label{eq:main_thm_prob_bound}
\end{equation}
In the following, we exploit a connection between determinants and perfect
matchings (P.M.'s) in bipartite graphs.
In section \ref{sec:our-construction}, we showed the correspondence of
the randomly constructed matrix $\mathbf{G}$ to an unbalanced bipartite graph
$G=(U,V,E)$.
The submatrix $\mathbf{G}_{\mathcal{S}}$ corresponds to a subgraph
$G_\mathcal{S}=(U,V_\mathcal{S},E_\mathcal{S})$,
depicted in Figure \ref{fig:bipartite-graph-random-unbalanced-subgraph}, where
$V_\mathcal{S}$ is a subset of $k^\prime$ vertices of
$V$, and $E_\mathcal{S}$ is a subset of the edges incident to vertices in
$V_\mathcal{S}$.
Similarly, a $k\times k$ submatrix $\mathbf{G}_{\mathcal{K}}$ of
$\mathbf{G}_{\mathcal{S}}$ corresponds to a smaller, balanced bipartite graph,
$G_{\mathcal{K}}=(U, V_{\mathcal{K}}, E_{\mathcal{K}})$, with
$k$ vertices on each side.
\begin{figure}[htbp!]
  \centering
  \begin{tikzpicture}[scale=0.9, ->,>=stealth',shorten >=1pt,auto,node
distance=1cm,
    thick,
     every node/.style={scale=0.9},
    symbol node/.style={circle, draw, minimum height=0.7cm},
    parity node/.style={rectangle, draw, minimum height=0.7cm, minimum
    width=0.7cm}]

    \node[symbol node] (U1) at (0,0) {};
    \node            (U2) [below of=U1] {$\vdots$};
    \node[symbol node] (U3) [below of=U2] {};
    \node[symbol node] (U4) [below of=U3] {};
    \node            (U5) [below of=U4] {$\vdots$};
    \node[symbol node] (U6) [below of=U5] {};

    \node[symbol node] (V1) at (3, 0.5) {};
    \node            (V2) [below of=V1] {$\vdots$};
    \node[symbol node] (V3) [below of=V2] {};
    \node[parity node] (V5) [below of=V3] {$+$};
    \node[node distance=1.5cm]            (V6) [below of=V5] {$\vdots$};
    \node[parity node,node distance=1.5cm] (V7) [below of=V6] {$+$};

    \draw (V1) edge [bend right, out=-30, in=-150] (U1);
    \draw (V3) edge [bend right, out=-30, in=-150] (U3);

    \draw (V5) edge [out=170, in=0]  (1.5, -2);
    \draw (V5) edge [out=180, in=22.5] (1.5, -2.5);
    \draw (V5) edge [out=190, in=45] (1.5, -3);

    \draw (V7) edge [out=160, in=-45] (1.5, -4.5);
    \draw (V7) edge [out=180, in=315, looseness=1] (1.5, -5);
    \draw (V7) edge [out=200, in=315, looseness=1] (1.5, -5.5);

    \draw [|<->|, gray] (-0.7,-2.35) -- (-0.7,0.35) node [midway, above, gray,
    sloped]{$s$};
    \draw [|<->|, gray] (-0.7,-5.35) -- (-0.7,-2.65) node [midway, above, gray,
    sloped]{$k-s$};

    \draw [|<->|, gray] (3.7, 0.85) -- (3.7,-1.85) node [midway, above, gray,
    sloped]{$s$};
    \draw [|<->|, gray] (3.7,-2.25) -- (3.7,-5.85) node [midway, above, gray,
    sloped]{$(1+\epsilon)k-s$};

    \path (-2.7, 0) (U1);

    \node[fit=(V1)(V7)](groupVs){};
    \draw[-, line width=.5pt,black,decorate,decoration={amplitude=7pt,brace}]
    ($(groupVs.north east)+(1, 0)$) -- ($(groupVs.south east)+(1, 0)$);
    \node[anchor=west, text width = 2cm] (vsname) at ($(V5)+(1.7, 0)$)
    {$V_{\mathcal{S}} \subseteq V$, $\left| V_{\mathcal{S}} \right| =
    k^\prime$.};
  \end{tikzpicture}
  \caption{%
  Bipartite Graph $G_\mathcal{S}=(U,V_\mathcal{S},
E_\mathcal{S})$ corresponding to $\mathbf{G}_\mathcal{S}$.
$V_\mathcal{S} \subseteq V$ consists of $k^\prime = (1+\epsilon)k$ vertices.
Out of them $s$ correspond systematic symbols, and $k^\prime - s$ to parities.
The graph has a perfect matching, \textit{i.e.} a matching saturating all $k$
vertices in $U$, with high probability, for any value of $s$.
  }
  \label{fig:bipartite-graph-random-unbalanced-subgraph}
\end{figure}
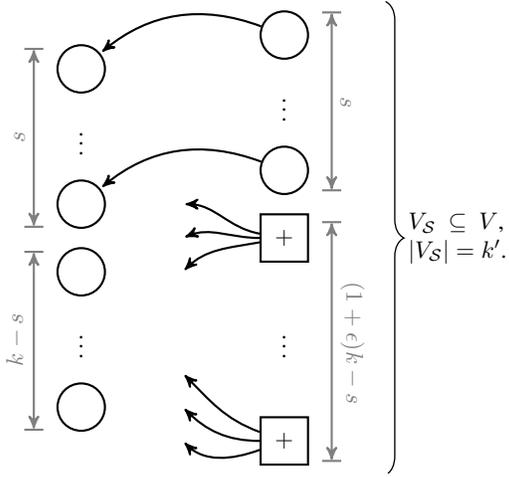

$\mathbf{G}_{\mathcal{K}}$ is closely related to the Edmond's matrix
$\mathbf{A}_{\mathcal{K}}$ of the
corresponding bipartite graph $G_{\mathcal{K}}$.
By definition, the $(i,j)$-th entry of $\mathbf{A}_{\mathcal{K}}$ is
\begin{flalign}
 A_{\mathcal{K}}(i,j) = \left\{ \begin{array}{ll}
                    a_{i,j}, & \text{if } (u_i, v_j) \in E_{\mathcal{K}}\\
		    0, 	& \text{if } (u_i, v_j) \notin E_{\mathcal{K}}
                   \end{array}
\right. ,
\end{flalign}
where $u_i \in U, v_j \in V_{\mathcal{K}}$, and $a_{i,j}$'s are
\textit{indeterminates}.
\begin{lemma}
  \label{lemma:connection-between-determinant-and-pm}
  The determinant of $\mathbf{A}_{\mathcal{K}}$ is nonzero if and only if
there
  exists a perfect matching in
  $G_{\mathcal{K}}$, \textit{i.e.},
  \begin{equation}
  \text{det}(\mathbf{A}_{\mathcal{K}})\neq 0  \quad \Leftrightarrow \quad
  \exists \text{ P.M. in } G_{\mathcal{K}}.
  \label{eq:equivalence-between-det-of-edmonds-and-pm-in-graph}
  \end{equation}
\end{lemma}
\noindent
However, $\mathbf{G}_{\mathcal{K}}$ is not an actual Edmond's matrix;
it is obtained from $\mathbf{A}_{\mathcal{K}}$ substituting the indeterminates
with randomly drawn elements of a finite field $\mathbb{F}_q$.
There are two substantially different cases in which
$\text{det}(\mathbf{G}_{\mathcal{K}}) =0$:
\begin{itemize}
 \item
 The determinant polynomial $\text{det}(\mathbf{A}_{\mathcal{K}})$ is
identically zero,  which by
Lemma \ref{lemma:connection-between-determinant-and-pm}
occurs if and only if $G_{\mathcal{K}}$ has no perfect matching, or
 \item
  it is not identically zero, \textit{i.e.}, $G_{\mathcal{K}}$ has a
perfect matching, but the selected coefficients correspond to a root of the
polynomial.
\end{itemize}
The equivalence property in
\eqref{eq:equivalence-between-det-of-edmonds-and-pm-in-graph} is not inherited
by $\mathbf{G}_{\mathcal{K}}$.
In contrast to the use of indeterminates, an unfortunate selection of the
random coefficients in $\mathbf{G}_{\mathcal{K}}$ can lead to zero determinant
even when $G_{\mathcal{K}}$ has a perfect matching.
However, if the coefficients are drawn from a sufficiently large field, the
probability of this event can be driven arbitrarily low.
More concretely, if $G_{\mathcal{K}}$ has a perfect matching $M$, then the
determinant of $\mathbf{A}_{\mathcal{K}}$ is a nonzero polynomial of degree
exactly $k$.
The probability that a random assignment of coefficients from
$\mathbb{F}_q$ yields a zero determinant can be upper bounded by
$k/q$, using the Schwartz-Zippel Lemma \cite{Motwani:2010}.
In summary,
\begin{align}
 \text{Pr}
 \left(
  \text{det}(\mathbf{G}_{\mathcal{K}})=0 \;|\; \exists
  \text{ P.M. in } G_{\mathcal{K}}
 \right)
 \le \frac{k}{q}.
 \label{eq:up-SZ-det-of-Gk}
\end{align}

The $k \times k^\prime $ matrix $\mathbf{G}_\mathcal{S}$ has
$\binom{k^\prime}{k}$ square submatrices.
For successful decoding it suffices that at least one such submatrix
$\hat{\mathbf{G}}_{\mathcal{K}}$ is nonsingular, \textit{i.e.}, has nonzero
determinant.
In light of \eqref{eq:up-SZ-det-of-Gk}, we ask instead
whether there exists a subgraph $\hat{G}_{\mathcal{K}}$ that has a P.M..
Observe that a P.M. in a subgraph $\hat{G}_{\mathcal{K}}$ is
also a P.M. in the larger graph $G_\mathcal{S}$.
Conversely, if $G_\mathcal{S}$ has a perfect matching $M$,
\textit{i.e.}, a matching saturating all vertices in $U$, then such a $k \times
k$ subgraph $\hat{G}_{\mathcal{K}}$ exists: its vertices are the endpoints of
the edges in $M$.
According to the following Lemma, $G_\mathcal{S}$ has a perfect matching $M$
with high probability.
\begin{lemma}
  \label{lemma:one_sided_pm_nonexistence}
  Consider the bipartite graph
  $G_\mathcal{S}=(U,V_\mathcal{S},E_\mathcal{S})$
  corresponding to the submatrix  $\mathbf{G}_{\mathcal{S}}$ of  $\mathbf{G}$.
  $V_\mathcal{S}$ contains any number $s$ of vertices with degree $1$ connected
to distinct vertices in $U$, and $(1\nobreak+\nobreak\epsilon)k-s$ vertices
that have randomly thrown $d(k) = c\log(k)$ edges as described in Section
\ref{sec:our-construction}.
  For appropriate constant $c \propto \rho/\epsilon $,
  \begin{align}
    \text{Pr}\left(\nexists \text{ P.M. in } G_\mathcal{S}
\right) \le 1/k^{\rho},
  \label{eq:upper-bound-on-prob-that-there-is-no-PM-in-G_S}
  \end{align}
for $\rho > 0$.
\end{lemma}
\noindent
The probability of equation \eqref{eq:main_thm_prob_bound} can be written as
\begin{flalign}
&\text{Pr}(\exists \; \mathcal{K} \subseteq \mathcal{S} :
\text{det}(\mathbf{G}_{\mathcal{K}})\neq 0 )  \nonumber\\
&= 1-\text{Pr}(\nexists  \; \mathcal{K} \subseteq \mathcal{S} :
\text{det}(\mathbf{G}_{\mathcal{K}})\neq 0 ) \nonumber\\
&=1- \left[ \underbrace{\text{Pr}(\nexists \; \mathcal{K} \subseteq \mathcal{S}
:
\text{det}(\mathbf{G}_{\mathcal{K}})\neq 0 | \exists M )}_{\alpha}\cdot
\text{Pr}( \exists M )\right. \nonumber \\
  &\qquad \;\;\;\left.+\underbrace{\text{Pr}(\nexists \; \mathcal{K} \subseteq
\mathcal{S} :
\text{det}(\mathbf{G}_{\mathcal{K}})\neq 0 | \nexists M )}_{\beta}\cdot
\text{Pr}( \nexists M ) \right]
\label{eq:probability_proof_1}.
\end{flalign}
As argued in the previous paragraph, the existence of $M$ implies the
existence of a subgraph $\hat{G}_{\mathcal{K}}$ that has a perfect matching.
The probability $\text{Pr}(\nexists \;
\mathbf{G}_{\mathcal{K}} : \text{det}(\mathbf{G}_{\mathcal{K}})\neq 0 \;|\;
\exists M )$, that all submatrices $\mathbf{G}_{\mathcal{K}}$ are singular
despite the existence of $M$, is upper bounded by the probability that
$\text{det}(\hat{\mathbf{G}}_{\mathcal{K}})= 0$, which was in turn upper
bounded in \eqref{eq:up-SZ-det-of-Gk} by $k/q$.
Hence, $\alpha \le k/q$.
On the other hand, nonexistence of a perfect matching in $G_{\mathcal{S}}$,
implies that no submatrix $\mathbf{G}_{\mathcal{K}}$ can have nonzero
determinant, hence, $\beta=1$.
Continuing from \eqref{eq:probability_proof_1}, we have:
\begin{flalign}
\text{Pr}\left(\exists \; \mathcal{K} \right. &
\left. \subseteq \mathcal{S}: \text{det}\left(\mathbf{G}_{\mathcal{K}}\right)\neq 0 \right) \nonumber \\
 & \ge
1-\left[ \frac{k}{q}\text{Pr}\left( \exists M \right)+ \text{Pr}\left( \nexists
M \right) \right] \nonumber\\
 & =
1-\frac{k}{q}\left(1-\text{Pr}\left( \nexists M \right)\right)-\text{Pr}\left(
\nexists M \right)\nonumber \\
 &= 
 1-\frac{k}{q}-\left(1-\frac{k}{q}\right)\text{Pr}( \nexists M ).
\label{eq:ub-on-prob-no-sing-square-submatrix}
\end{flalign}
Finally, satisfying the conditions in Lemma
\ref{lemma:one_sided_pm_nonexistence}, we can guarantee that $\text{Pr}(
\nexists M ) \le k^{-\rho}$, for $\rho > 1$.
Applying the bound on the right hand side of  
\eqref{eq:ub-on-prob-no-sing-square-submatrix}, we
obtain the desired result in \eqref{eq:main_thm_prob_bound}.
To complete the proof, it remains to prove Lemmata
\ref{lemma:connection-between-determinant-and-pm} and
\ref{lemma:one_sided_pm_nonexistence}.

\subsubsection{Proof of Lemma \ref{lemma:connection-between-determinant-and-pm}
- Connection between determinants and perfect matchings}

We use the following expression for the determinant:
\begin{equation}
 \text{det}(\mathbf{A}_{\mathcal{K}}) = \sum_{\pi \in S_n}{\text{sgn}(\pi)}
\prod_{i=1}^{n}\mathbf{A}_{\mathcal{K}}(i, \pi(i)) \label{eq:determinant},
\end{equation}
where $S_n$ is the set of all permutations on $\{1, \hdots, n\}$, and
$\text{sgn}(\pi)$ is the sign of permutation $\pi$.
There is a one-to-one correspondence between a permutation $\pi \in S_n$ and a
candidate perfect matching $\left\{ (u_1, v_{\pi(1)}), \hdots, (u_n, v_{\pi(n)})
\right\}$ in $G_{\mathcal{K}}$.
Note that if the candidate P.M. does not exist in $G_{\mathcal{K}}$,
\textit{i.e.}, some edge $(u_i, v_{\pi(i)}) \notin E_{\mathcal{K}}$ then the
term corresponding to $\pi$ in the summation is $0$.
Therefore,
\begin{equation}
 \text{det}(\mathbf{A}_{\mathcal{K}}) = \sum_{\pi \in
\mathcal{P}}{\text{sgn}(\pi)} \prod_{i=1}^{n}a_{i, \pi(i)},
\end{equation}
where $\mathcal{P}$ is the set of perfect matchings in
$G_{\mathcal{K}}$.
If $\mathcal{P}=\emptyset$, \textit{i.e.}, if $G_{\mathcal{K}}$ has no P.M.,
every term in the sum is equal to zero.
If on the contrary $G_{\mathcal{K}}$ has a P.M., there exists a $\hat{\pi} \in
\mathcal{P}$, and hence the term corresponding to $\hat{\pi}$ is
$\prod_{i=1}^{n}a_{i,
\hat{\pi}(i)} \neq 0$.
Additionally, there is no other term in the summation containing the exact same
set of variables and this term cannot be canceled out.
In this case, $\text{det}(\mathbf{A}_{\mathcal{K}}) \neq 0$, which concludes
the proof of the lemma. $\qed$
\subsubsection{Proof of Lemma \ref{lemma:one_sided_pm_nonexistence}: Existence
of Perfect Matching in the random subgraph}
\label{sec:ospmorubg}
We want to establish an upper bound on the probability that there is no
\textit{perfect matching (P.M.)} between $U$ and $V_\mathcal{S}$ in the random
$k \times k^\prime$ bipartite graph $G_\mathcal{S}$.
In fact, we want to show that $d(k)=O(\log{k})$ in the construction of the
bipartite graph, suffices to achieve an upper bound asymptotically decaying with
a rate $1/\text{poly}(k)$.

Let $V_{s}$, $0\le |V_s| \le k$, denote the subset of $V_\mathcal{S}$
corresponding to systematic
encoded symbols.
If a P.M. exists, we may assume that all symbols in $V_s$ participate in it.
To see that, consider a vertex $v_i \in V_s$, connected to a symbol $u_i \in
U$, and assume that $(v_i, u_i)$ is \textit{not} included in the P.M.
Then, $u_i$ must be paired with some vertex $v_j \notin V_s$, since $v_i$ was
the only systematic symbol connected to $u_i$.
In addition, $u_i$ is the only symbol adjacent to $v_i$, hence $v_i$ does not
participate in the P.M.
Given such a P.M., we can construct another one substituting $(u_i, v_j)$
with $(u_i, v_i)$.
Therefore, without loss of generality, we may assume that all vertices in
$V_{s}$ participate in the P.M.

Since $k^\prime=(1+\epsilon)k>k$, $V_\mathcal{S}$ contains a nonempty
subset corresponding to parity symbols, denoted by
$\overline{V}_{s}=V_\mathcal{S}\backslash V_s$.
Let $U_{s}$ denote the subsets of $U$ matched with vertices in $V_s$, and
$\overline{U}_{s}=U\backslash U_s$ the remaining vertices.
A P.M. between  $U$ and $V_\mathcal{S}$ exists if and only if a P.M.
exists between $\overline{U}_{s}$ and $\overline{V}_{s}$.

The probability that a P.M. does not exist equals the probability that there
exists  a \textit{contracting} set of vertices in $\overline{U}_{s}$,
\textit{i.e.}, a subset of $\overline{U}_{s}$ with a joint neighborhood
smaller than its cardinality.
Let $|V_{s}| = s$ and $|\overline{V}_{s}|=k^\prime-s$.
Denote by $E_i$ the event that there exists a set of $i$ vertices in
$\overline{U}_{s}$ that 
\textit{contracts}, \textit{i.e.},  has at most $i-1$ neighbors in
$\overline{V}_{s}$.
This is equivalent to at least  $k^\prime-s-(i-1)$ vertices in
$\overline{V}_{s}$ being
only
adjacent to vertices in $\overline{U}_{s}$ other than the $i$ vertices of
interest. Then,
\begin{flalign}
     &\text{Pr}\left(\nexists P.M. \text{ in } G_{\mathcal{S}} \right)
=  \text{Pr}(\bigcup_{i=1}^{|\overline{U}_{s}|}E_i )
\le  \sum_{i=1}^{|\overline{U}_{s}|}\text{Pr}(E_i )\nonumber\\
&\le
\resizebox{0.9\hsize}{!}{$
\displaystyle
\sum_{i=1}^{|\overline{U}_{s}|}\binom{|\overline{U}_{s}|}{i}\binom{|\overline{V}
_{s}|}{| \overline{V}_{s} |-(i-1)}\left(
\frac{k-i}{k}\right)^{d(k) \left(|\overline{V}_{s}|-(i-1)\right)}
$}\nonumber \\
&=
\resizebox{0.9\hsize}{!}{$
\displaystyle
\sum_{i=1}^{k-s}\binom{k-s}{i}\binom{k^\prime-s}{
k^\prime-s-i+1}\left(
\frac{k-i}{k}\right)^{d(k) \left(k^\prime-s-i+1\right)}
$}\nonumber \\
&=
\resizebox{0.9\hsize}{!}{$
\displaystyle
\underbrace{
\sum_{i=1}^{k-s}\binom{k-s}{k-s-i}\binom{k^\prime-s}{k^\prime-s-i+1}\left(
\frac{k-i}{k}\right)^{d(k) \left(k^\prime-s-i+1\right)} }_{A}
$}.
\nonumber
\end{flalign}
\begin{lemma}
\label{lemma:bound-on-binomial-coefficient}
The binomial coefficient satisfies the well-known bound
\begin{flalign}
\binom{n}{k}\le 2^{n H\left( \frac{k}{n}\right)},
\label{eq:ub-on-binom-coef}
\end{flalign}
where $H(p)=p\log_2 p + (1-p)\log_2(1-p)$ is the binary entropy function.
\end{lemma}
\begin{proof}
From the probability mass function of the binomial distribution with $n$
trials and probability of success $p=\frac{k}{n}<1$, we have
  \begin{flalign}
    1 &= \sum_{m=0}^{n} \binom{n}{m} p^m(1-p)^{n-m}
    \ge \binom{n}{k} p^k(1-p)^{n-k} \nonumber \\
    &= \binom{n}{k} 2^{n\left[ p\log_2{p}+(1-p)\log_2{(1-p)}\right]}
    = \binom{n}{k} 2^{-nH(p)}. \nonumber
  \end{flalign}
\end{proof}
\noindent
Applying \eqref{eq:ub-on-binom-coef} on the coefficients of $A$, we obtain
\begin{align}
 A &\le
 \sum_{i=1}^{k-s}
 2^{\left[ B_{1}(i)+B_{2}(i)+B_{3}(i)\right]}
\nonumber \\
 &\le (k-s) \max_{i}\left( 2^{\left[ B_{1}(i)+B_{2}(i)
+B_{3}(i)\right]}  \right),
\label{eq:ub-on-A-with-max}
\end{align}
where
\begin{itemize}
 \item $B_{1}(i)={(k-s)H\left(\frac{k-s-i}{k-s} \right)}$,
 \item $B_{2}(i)=$ $(k^\prime-s)$ $H\left(\frac{k^\prime-s-i+1}
{k^ {'} -s}\right)$, and
 \item $B_{3}(i) = d(k) \left( k^\prime-s-i+1\right)\log_2\left(
\frac{k-i}{k}\right)$.
\end{itemize}
Towards our objective, it suffices to require the right hand side of
\eqref{eq:ub-on-A-with-max} to vanish asymptotically faster than
$1/k^\rho, \rho>0$, for each value of $s \in \{0, \hdots, k\}$.
Equivalently, it suffices
\begin{flalign}
& \log_2{(k-s)}+\sum_{l=1}^{3} B_{l}(i) \le -\rho \log_2{(k)}
\label{eq:log-bound-on-A},
%
\end{flalign}
for all $ 0 \le s \le k$ and $ 1 \le i \le k-s$.
Expanding and rearranging terms we find that in order for
\eqref{eq:log-bound-on-A} to hold, it suffices
\begin{flalign}
 d(k) \ge
\frac{%
\resizebox{0.73\hsize}{!}{$
\left[
\begin{array}{l}
 \rho\log_2{(k)}+2\log_2(k-s)+ \log_2\left( k^\prime-s \right) \\+(k-s) H \left(
\frac{k-s-i}{k-s} \right) + (k^\prime-s) H \left(
\frac{k^\prime-s-i+1}{k^\prime-s}
\right)
\end{array}
\right]
$}
}{- \left( k^\prime-s-i+1 \right)\log_2\left( \frac{k-i}{k}\right)}.
\label{eq:random_almost_mds_inequality_on_d}
\end{flalign}

Our objective is now to show that the right hand side of
\eqref{eq:random_almost_mds_inequality_on_d} is $O(\log k)$.
Let $N$ and $D$ denote the numerator and denominator of the right hand side of
inequality \eqref{eq:random_almost_mds_inequality_on_d}.
For the numerator $N$, we have the following upper bound:
\begin{flalign}
 N& \le
%
\underbrace{(3+\rho)\log_2{(k)}+\log_2{(1+\epsilon)}}_{N_1} \nonumber\\
&\quad+\underbrace{k H \left( \frac{k-i}{k} \right)}_{N_2} +
\underbrace{(1+\epsilon)k H
\left(
\frac{(1+\epsilon)k-i+1}{(1+\epsilon)k}
\right)}_{N_3},
\nonumber
&
\end{flalign}
where the inequality is due to the monotonicity of the logarithm and the fact
that $g(x)=xH\left( \frac{x-y}{x} \right)$ is  increasing with respect to $x$
for $0 \le y \le x$.
For the denominator $D$, we have:
\begin{flalign}
\resizebox{0.99\hsize}{!}{$
D
 = \underbrace{\left( k-s-i \right. }_{\ge 0}+ \left. \epsilon k+1 \right)
\underbrace{\log_2\left( \frac{k}{k-i}\right)}_{\ge 0}
\ge  \epsilon k \log_2\left( \frac{k}{k-i}\right). \nonumber
$}
\end{flalign}
Recall that $\log(1+x)>x / (x+1)$ for $x>-1, x\neq0$.
Applying the inequality for $x=\frac{i}{k-i}>0$, we find that $D$ can be
further lower bounded as follows 
\begin{align}
 D \ge \epsilon k \frac{i}{k \log 2} \le \frac{\epsilon}{\log 2}i.
 \label{eq:llb-on-D}
\end{align}
We examine the ratio $N / D$ in parts.

\begin{enumerate}[(i)]
 \item For the first part, and for $k \ge (1+\epsilon)$ we have
\begin{flalign}
 \frac{N_1}{D}
 & \le \frac{(3+\rho+1)\log_2{(k)}}{
\epsilon k
\log_2\left( \frac{k}{k-i}\right)}
\le
\frac{(4+\rho)}{\epsilon} \log(k),
&
\label{eq:ratio_bound_pt1}
\end{flalign}
where for the second inequality we have used \eqref{eq:llb-on-D} and the fact
that $i \ge 1$.
\item For the second part, expanding the entropy we have
\begin{flalign}
\frac{N_2}{D}
&\le \frac{ k H \left( \frac{k-i}{k} \right)}{\epsilon k \log_2{\left(
\frac{k}{k-i}\right)} }
 = \frac{k-i}{\epsilon k}-\frac{\frac{i}{k}\log_2\left(
\frac{i}{k}\right)}{\epsilon  \log_2{\left( \frac{k}{k-i}\right)} }
\nonumber \\
&\stackrel{\eqref{eq:llb-on-D}}{\le}  \frac{1}{\epsilon}+\frac{1}{\epsilon}\frac{i \log_2\left(
\frac{k}{i}\right) }{
i / \log 2 }
\le \frac{2}{\epsilon} \log\left( k \right),
&
\label{eq:ratio_bound_pt2}
\end{flalign}
where the last inequality holds for $k\ge 2$.
\item
For the third part, for $i=1$, $N_3/D = 0$.
For $i \ge 2$, first observe that $h(k)=\log\left(\frac{k}{k-i} \right)$
is decreasing in $k$ for $0 \le i \le k$.
Since $(1+\epsilon)k > k$ and $ 2 \le i \le k-s$, exploiting the
monotonicity of $h(k)$, we have
\begin{flalign}
\frac{N_3}{D}
& \le \frac{ (1+\epsilon) k H \left(
\frac{(1+\epsilon)k-i+1}{(1+\epsilon)k} \right)}{\epsilon k \log_2{\left(
\frac{(1+\epsilon)k}{(1+\epsilon)k-i}\right)} }
\nonumber\\
&
\stackrel{\eqref{eq:llb-on-D}}{\le} \frac{(1+\epsilon)}{\epsilon}
	 \left(
	    1+ \frac{i-1}{(1+\epsilon)k}
	    \frac{\log_2{\frac{(1+\epsilon)k}{i-1}}}
	    {\displaystyle\frac{i}{(1+\epsilon)k \log 2} }\right)\nonumber\\
& \le \frac{2(1+\epsilon)}{\epsilon} \log \left(k\right),
&
\label{eq:ratio_bound_pt3}
\end{flalign}
where the last inequality holds when $k > e(1+\epsilon)$.
\end{enumerate}
Combining \eqref{eq:ratio_bound_pt1}, \eqref{eq:ratio_bound_pt2} and
\eqref{eq:ratio_bound_pt3},
we conclude that using $d(k)=c\log(k)$, where $c  = (8 + \rho +
2\epsilon)/\epsilon$,
suffices to force $\text{Pr}(\nexists \text{P.M.} \text{ in }
G_\mathcal{S}) \le 1/k^{\rho}$, which completes the proof. $\qed$


\subsection{Proof of Theorem \ref{thm:main-theorem:converse}}

Consider the decoding graph $G_\mathcal{S}$
corresponding to the $k \times k^\prime$ submatrix $\mathbf{G}_{\mathcal{S}}$ of  $\mathbf{G}$.
$G_\mathcal{S}$ is a random bipartite graph between $k$ input and $k^\prime$ encoded nodes, such that each encoded node has degree at most $d(k)$.

An input symbol is \textit{covered} by the set of $k^\prime$ encoded symbols, 
if and only if it participates with a nonzero coefficient in the formation of at least one symbol in the set.
In terms of the decoding graph, an input node is covered if and only if it is adjacent to at least one encoded node.

The probability of decoding failure is lower bounded by the probability that an uncovered input node exists in $G_\mathcal{S}$:
all input nodes being covered is a prerequisite for the $k$ input symbols to be retrievable from a set of $k^\prime$ encoded symbols.

The problem is equivalent to throwing $k^\prime \cdot d(k)$ balls into $k$ bins and
requiring that no bin is empty with high probability.
It is a standard result in balls and bins analysis that  throwing $\Omega
\left(k\log{k}\right)$ balls is necessary to that end.
It is hence imperative that $k^\prime \cdot d(k) =\Omega
\left(k\log{k} \right)$.
Taking into account that $k^\prime = (1+\epsilon)k$, we obtain the desired
result. $\qed$

\subsection{Proof of Theorem
\ref{thm:ub-prob-exist-u-few-parities}}

Let $\mathcal{C}_v \subseteq [k]$ be the subset of systematic symbols covered
by a parity $v$.
Also, let $\mathcal{P}_u$ be the set of parity symbols covering a systematic
symbol
$u$, and $M = \left| \mathcal{P}_u \right|$.
Note that $u \in \mathcal{C}_v \Leftrightarrow v \in \mathcal{P}_u$.
The total number of generated parities is $rk$.
Hence, $M$ is a binomial random variable with $rk$ trials and probability of
success equal to $\text{Pr}\left( u \in \mathcal{C}_v \right)$, the probability
that a parity $v$ covers the systematic symbol $u$.

Every parity $v$ throws its $d(k)$ edges uniformly at random over $[k]$,
independently, with replacement.
A simple union bound yields
\begin{align}
  \text{Pr}\left( u \in \mathcal{C}_v
  \right) 
  \le \sum_{l=1}^{d(k)}\frac{1}{k}
  = \frac{d(k)}{k}.
  \label{eq:ub-prob-v-covers-u}
\end{align}
Similarly, we can obtain a lower bound:
\begin{align}
  \text{Pr}\left(  u \in \mathcal{C}_v \right)
  &= 1 - \text{Pr}\left( u \notin \mathcal{C}_v \right)
  = 1- \left( 1 - \frac{1}{k}\right)^{d(k)} \nonumber \\
  &\ge 1- \exp\left( -\frac{d(k)}{k} \right).
  \label{eq:lb-prob-v-covers-u-using-exp}
\end{align}
Using the fact that $\exp(x) \le 1+x+x^2/2$ for $x \ge 0$, inequality
\eqref{eq:lb-prob-v-covers-u-using-exp} is simplified into
\begin{align}
  \text{Pr}\left(  u \in P_v \right) \ge \frac{d(k)}{k}-\frac{d(k)^2}{k^2}.
  \label{eq:lb-prob-v-covers-u}
\end{align}
The probability of that parity $v$ covers $u$ lies in the range described by
\eqref{eq:ub-prob-v-covers-u} and
\eqref{eq:lb-prob-v-covers-u}.
Based on these bounds, we can calculate a range  for the expected value of $M$.
Taking into account that $d(k) = c\log{k}$, we have
\begin{align}
rc\log(k) - \frac{rc\log^2(k)}{k} \le E \left[ M \right]  \le rc\log(k).
\label{eq:bound-on-expected-M}
\end{align}
Since each parity is created independently,
the following Chernoff bound on the lower tail of the distribution of
$M$ holds:
\begin{align}
 \text{Pr}\left( M \le (1-\epsilon)E[M] \right)  &\le \exp\left(
-\frac{\epsilon^2}{2} E[M] \right).
 \label{eq:chernoff-ub-lowtail-expected-number-of-cover} 
\end{align}
The right hand side of
\eqref{eq:chernoff-ub-lowtail-expected-number-of-cover} can be
further bounded as follows
\begin{align}
\text{Pr}\left( M \right. &\le \left. (1-\epsilon)E[M] \right) \nonumber\\
&\le 	\exp\left( -\frac{\epsilon^2}{2} \left[  rc\log(k) -
\frac{rc\log^2(k)}{k} \right]\right) \nonumber \\
&= 	\exp\left( - \log(k^{rc\epsilon^2/2}) \right)
\exp\left( \frac{\epsilon^2}{2} \frac{rc\log^2(k)}{k} \right)  \nonumber \\
&\le \frac{1}{k^{rc\epsilon^2/2}}\exp\left(
 \frac{\epsilon^2rc\log^2(k)}{2k} \right),
\end{align}
which is the desired result. $\qed$
\subsection{Proof of Theorem \ref{thm:availability-of-the-code}}
Let $\mathcal{P}_u$  denote the set of parity symbols that cover systematic
symbol $u$.
We assume for simplicity that $\left| \mathcal{P}_u \right| = m =
rc\log(k)$.
Given that every parity was generated independently, the footprints
$\mathcal{F}_{v}$ for $v \in \mathcal{P}_u$ are independent random variables.
We are interested in the maximum cardinality subset $\mathcal{D} \subseteq
\mathcal{P}_u$ such that any two parity symbols  $v_i, v_j \in \mathcal{D}$ are
isolated.

Consider the graph $H_u = (\mathcal{P}_u, E)$, where
edge $(v_i,v_j) \in E$ if and only if $v_i$ and $v_j$ are not isolated.
Then, $\mathcal{D}$ corresponds to the maximum independent set in $H_u$,
$I(H_u)$.
Since $H_u$ is a random graph, its independence number $\alpha(H_u)$ is a
random variable $Z
=f\left(\mathcal{F}_{v_1},\hdots,\mathcal{F}_{v_m}\right) $, which is a function of the
$m$ independently
drawn $\mathcal{F}_{v_i}$'s.
Function $f(\cdot)$ satisfies the bounded differences
condition, \textit{i.e.}, for any configuration $\mathcal{F}_{v_1}, \hdots,
\mathcal{F}_{v_m}$, substituting a single variable $\mathcal{F}_{v_i}$ with
another variable $\mathcal{F}_{v_i}^\prime$ cannot impact the function value
arbitrarily.
In the graph analogy, substituting $\mathcal{F}_{v_i}$ with
$\mathcal{F}_{v_i}^\prime$ for some $i$,
corresponds to removing a vertex from $H_u$ along with its incident edges, and
inserting a new vertex arbitrarily connected to other vertices.
\begin{lemma}
\label{lemma:size-of-independent-set-in-graph}
Consider an undirected graph $H = (V,E)$ with $|V|=M$ and let $\alpha(H)$ denote
its independence number, \textit{i.e.} the cardinality of the maximum
independent set $I(H) \subseteq V$.
Construct a graph $H^\prime$ as follows:
remove a node $v$ from $H$ along with all incident edges and
insert a new node $v^\prime$ connected to an arbitrary set of vertices in $H$.
Then $|\alpha(H) - \alpha(H^\prime)| \le 1$.
\end{lemma}
\begin{proof}
 Regardless of whether $v \in I(H)$ or not, the set $\mathcal{S} =
I(H)\backslash \{v\}$ is common in $H$ and $H^\prime$ and remains an independent
set in the latter.
 Therefore, $\alpha(H^\prime) \ge  |\mathcal{S}| \ge \alpha(H)-1$, where
equality in the second inequality holds only if $v \in I(H)$.
 Inversely, consider the maximum independent set in $H^\prime$,
$I(H^\prime)$.
 Irrespectively of whether $v^\prime \in I(H^\prime)$, the set
$\mathcal{S}^\prime = I(H^\prime)\backslash \{v^\prime\}$ is an independent set
 in $H$ too, implying that $\alpha(H) \ge \alpha(H^\prime)-1$.
 We conclude that $-1 \le \alpha(H) -\alpha(H^\prime) \le 1$, which is the
desired result.
\end{proof}
Based on the previous lemma, we have
\begin{align}
\scalebox{0.96}{$
 \displaystyle
 \max_{\mathcal{F}_{v_1},\hdots, \mathcal{F}_{v_m}, \mathcal{F}_{v_i}^\prime}\left|
 f\left(\hdots,\mathcal{F}_{v_i},\hdots \right) -
 f\left(\hdots,\mathcal{F}_{v_i}^\prime,\hdots\right)
 \right| \le 1.$
 }
 \label{eq:bounded-diff-on-function-of-footprints}
 \end{align}
 Provided that $\mathcal{F}_{v_i}$'s are independent and
$f(\cdot)$ satisfies condition
\eqref{eq:bounded-diff-on-function-of-footprints}, 
McDiarmid's inequality \cite{mcdiarmid:boundeddiff} yields
\begin{align}
 \text{Pr}\left(Z  \le E[Z]-t\right)\;
 \le
 \exp\left(- \frac{2 t^2}{m} \right),
 \label{eq:mcdiarmid-bound-on-maximum-independent-set}
\end{align}
for $t > 0$.

The concentration result of
\eqref{eq:mcdiarmid-bound-on-maximum-independent-set} holds, even if $E[Z]$
remains unknown.
A trivial lower bound on $E[Z]$ can be obtained using those vertices in $H_u$
that are disconnected components, \textit{i.e.}, that have degree equal to
zero.
Such vertices correspond to parity symbols that are isolated from all
other symbols in $\mathcal{P}_u$, not only those in $\mathcal{D}$,
and are always members of the maximum independent set.
The probability that a parity symbol $v \in \mathcal{P}_u$ is isolated from all
other symbols in $\mathcal{P}_u$ is
\begin{align}
\text{Pr}\left( v \text{ is isolated} \right)
&=
 \left(   \frac{k-\left| \cup_{j=1, j\neq i}^{m-1} \mathcal{F}_{v_j}\right|}{k}
\right)^{d(k)-1} \nonumber\\
 &\ge
 \left(  1 - \frac{(m-1)(d(k)-1)}{k} \right)^{d(k)-1}  \nonumber \\
 & = \exp\left( -\frac{(m-1)(d(k)-1)^2}{k} \right) \nonumber \\
 &\ge 1 -\frac{(m-1)(d(k)-1)^2}{k} \nonumber\\
 & \ge 1 -\frac{md(k)^2}{k}.
 \label{eq:lower-bound-on-probability-parity-v-is-isolated}
\end{align}
Multiplying with $m$, the number of symbols in $\mathcal{P}_u$, we obtain a
lower bound on the expected number of completely isolated parities in
$\mathcal{P}_u$, which in turn is a lower bound on $E[Z]$.
In other words,
\begin{align}
 m - m^2\frac{d(k)^2}{k} \le E[Z],
\end{align}
Therefore, we have
\begin{align}
 \text{Pr}\left(Z  \le m -  \frac{ \left[ m d(k) \right]^2}{k}-t\right)
 \le
  \text{Pr}\left(Z  \le E\left[Z\right]-t\right).
  \nonumber
\end{align}
Combining the above, with inequality
\eqref{eq:mcdiarmid-bound-on-maximum-independent-set}, we conclude that
\begin{align}
  \text{Pr}\left(Z  \le m -  \frac{ \left[ m d(k) \right]^2}{k}-t\right)
 \le
 \exp\left(- \frac{2 t^2}{m} \right).
 \label{eq:mcdiarmid-bound-on-the-lower-bound-of-expected-Z}
\end{align}
Let $t = (1-\alpha)m - \frac{\left[md(k)\right]^2}{k}$ for some $\alpha \in
(0,1)$. For sufficiently large $k$, $t$ will be nonnegative.
Substituting $t$ in \eqref{eq:mcdiarmid-bound-on-the-lower-bound-of-expected-Z},
we obtain
\begin{align}
 \text{Pr}\left( Z \le  \alpha m \right)
 \le
 \exp\left( - \frac{2\left[(1-\alpha)m- \frac{\left[md(k)\right]^2}{k}
\right]^2}{m} \right).
 \label{eq:mcdiarmids-bound-for-Z-less-than-alogk}
\end{align}
We are interested in the value of $\alpha$ for which $\text{Pr}\left( Z \le
\alpha m \right)$ decreases faster than $1/k^{1+\lambda}$, for some $\lambda >0$.
It suffices to require the right hand side
of \eqref{eq:mcdiarmids-bound-for-Z-less-than-alogk} is less than
$1/k^{1+\lambda}$.
Taking logarithms on both sides, it suffices to find $\alpha$ such that
\begin{align}
%
&
- \frac{2\left[(1-\alpha)m- \frac{m^2d(k)^2}{k} \right]^2}{m}
\le
-(1+\lambda) \log(k)  \nonumber\\
\Leftrightarrow &
- 2\left[{\underbrace{(1-\alpha)}_{\overline{\alpha}}}^2m- \frac{m^2d(k)^2}{k}
\right]^2
\le
-(1+\lambda) m\log(k) \nonumber\\
\Leftrightarrow &
\overline{\alpha}^2m^2 - \overline{\alpha} 2\frac{m^3d(k)^2}{k} +
\frac{m^4d(k)^4}{k^2} - \frac{1+\lambda}{2}m\log(k)
\ge
0. \nonumber
\end{align}
The last inequality is a quadratic inequality on $\overline{\alpha}$, satisfied
when
\begin{align}
 \overline{\alpha}
 &\ge
 \frac{2\frac{m^3d^2}{k}+ \sqrt{4 m^2\frac{1+\lambda}{2}m\log(k)}}{2
m^2}\nonumber\\
 &= \frac{m d(k)^2}{k}+ \sqrt{\frac{1+\lambda}{2}\frac{\log(k)}{m}}\nonumber\\
 &= \frac{rc^3\log^3(k) }{k}+ \sqrt{\frac{1+\lambda}{2rc}},
\end{align}
which for appropriate choice of $r$ and $c$, can be a solution
that asymptotically lies in $(0,1)$.
Therefore, for
\begin{align}
 \alpha
 \le
 1 - \sqrt{\frac{1+\lambda}{2rc}} - \frac{rc^3\log^3(k) }{k},
\label{eq:bound-for-valid-alpha}
 \end{align}
we have
\begin{align}
 \text{Pr}\left( Z \le \alpha \cdot m \right) \le k^{-(1+\lambda)}.
 \label{eq:bound-we-want-for-the-lower-tail-of-Z}
\end{align}
The probability that there exists a systematic symbol with availability lower
than $\alpha m$ can be bounded with a union bound over all systematic
symbols:
\begin{align}
 \text{Pr}\left( \exists u \text{ not $(\alpha m)$-available}\right) \le k \cdot k^{-(1+\lambda)} = k^{-\lambda},
\end{align}
which completes the proof.
$\qed$

\bibliographystyle{IEEEtran}

\end{document}